\documentclass[smallextended]{svjour3}       % onecolumn (second format)
\smartqed  % flush right qed marks, e.g. at end of 
\usepackage{graphicx}
 \journalname{Designs, Codes and Cryptography}

\usepackage{amsmath}
\usepackage{amsfonts}
\usepackage{amssymb}
\usepackage{amsthm}
\usepackage{bbm}
\usepackage{float}
\usepackage{makecell}
\usepackage{url}
\usepackage{amscd}
\usepackage{mathrsfs}

\usepackage{setspace}
\usepackage[OT2,T1]{fontenc}
\DeclareSymbolFont{cyrletters}{OT2}{wncyr}{m}{n}
\DeclareMathSymbol{\Sha}{\mathalpha}{cyrletters}{"58}

\newcommand{\ep}{\varepsilon}
\newcommand{\E}{\mathbb{E}}
\renewcommand{\H}{\mathcal{H}}
\newcommand{\R}{\mathcal{R}}
\newcommand{\C}{\mathcal{C}}

\newcommand{\F}{\mathbb{F}}
\renewcommand{\P}{\mathbb{P}}

\newcommand{\rank}{\mathrm{rk}}
\renewcommand{\l}{\ell}
\renewcommand{\d}{\mathrm{d}}

\newcommand{\Cov}{\mathrm{Cov}}
\newcommand{\Gaubin}{\genfrac{[}{]}{0pt}{}}

\begin{document}

\title{Decoding error probability of random parity-check matrix ensemble over the erasure channel \thanks{The research of Fang-Wei Fu was supported in part by the National Key Research and Development Program of China under Grant 2018YFA0704703, in part by the National Natural Science Foundation of China under Grant 61971243, in part by the Natural Science Foundation of Tianjin under Grant 20JCZDJC00610, and in part by the Nankai Zhide Foundation. The research of M. Xiong was supported by RGC grant number 16306520 from Hong Kong.
}%\thanks{Grants or other notes
}

\titlerunning{Decoding error of random PC matrix ensemble over EC}        % if too long for running head

\author{Chin Hei Chan \and Fang-Wei Fu \and Maosheng Xiong %etc.
}


\institute{Chin Hei Chan \at
              Mathematics Department, Hong Kong University of Science and Technology, Clear Water Bay, Hong Kong\\
              %Tel.: \\
              \email{chchanam@connect.ust.hk}           %  \\
            \and
            Fang-Wei Fu \at
            Chern Institute of Mathematics and LPMC, and the Tianjin Key Laboratory of Network and Data Security Technology, Nankai University, Tianjin, 300071, China \\
            \email{fwfu@nankai.edu.cn}
           \and
           Maosheng Xiong \at
            Mathematics Department, Hong Kong University of Science and Technology, Clear Water Bay, Hong Kong\\              \email{mamsxiong@ust.hk}          
    }

\date{Received: date / Accepted: date}

\maketitle
\begin{abstract}
Using tools developed in a work by Shen and the second author, in this paper we carry out an in-depth study on the average decoding error probability of the random parity-check matrix ensemble over the erasure channel under three decoding principles, namely unambiguous decoding, maximum likelihood decoding and list decoding. We obtain explicit formulas for the average decoding error probabilities of the random parity-check matrix ensemble under these three decoding principles and compute the error exponents. Moreover, for unambiguous decoding, we compute the variance of the decoding error probability of the random parity-check matrix ensemble and the error exponent of the variance, which implies a strong concentration result, that is, roughly speaking, the ratio of the decoding error probability of a random linear code in the ensemble and the average decoding error probability of the ensemble converges to 1 with high probability when the code length goes to infinity.

\end{abstract}

\section{Introduction}

\subsection{Background}

In digital communication, it is common that messages transmitted through a public channel may be distorted by the channel noise. The theory of error-correcting codes is the study of mechanisms to cope with this problem. This is an important research area with many applications in modern life. For example, error-correcting codes are widely employed in cell phones to correct errors arising from fading noise during high frequency radio transmission. One of the major challenges in coding theory remains to construct new error-correcting codes with good properties and to study their decoding and encoding algorithms.

In a binary erasure channel (BEC), a binary symbol is either received correctly or totally erased with probability $\ep$. The concept of BEC was first introduced by Elias in 1955 \cite{InfThe}. Together with the binary symmetric channel (BSC), they are frequently used in coding theory and information theory because they are among the simplest channel models, and many problems in communication theory can be reduced to problems in a BEC. Here we consider more generally a $q$-ary erasure channel in which a $q$-ary symbol is either received correctly, or totally erased with probability $\ep$.

The problem of decoding linear codes over the erasure channel has received renewed attention in recent years due to their wide application in the internet and the distributed storage system in analyzing random packet losses \cite{Byers,Luby,Lun}. Three important decoding principles, namely \emph{unambiguous decoding}, \emph{maximum likelihood decoding} and \emph{list decoding}, were studied in recent years for linear codes over the erasure channel, the corresponding decoding error probabilities under these principles were also investigated (see \cite{Didier,Lemes,FFW,Weber} and reference therein). %

In particular in \cite{FFW}, upon improving previous results, the authors provided a detailed study on the decoding error probabilities of a general $q$-ary linear code over the erasure channel under the three decoding principles. Via the notion of $q^\l$-incorrigible sets for linear codes, they showed that all these decoding error probabilities can be expressed explicitly by the $r$-th support weight distribution of the linear codes. As applications they obtained explicit formulas of the decoding error probabilities for some of the most interesting linear codes such as MDS codes, the binary Golay code, the simplex codes and the first-order Reed-Muller codes etc. where the $r$-support weight distributions were known. They also computed the average decoding error probabilities of a random $[n,k]_q$ code over the erasure channel and obtained the error exponent of a random $[n,nR]_q$ code ($0<R<1$) under one of the decoding principles. The error exponents of a random $[n,nR]$ code under the other two decoding principles were obtained in \cite{Xio}.

\subsection{Statement of the main results}

In this paper we consider a different code ensemble, namely the \emph{random parity-check matrix ensemble} $\R_{m,n}$, that is, the set of all $m \times n$ matrices over $\F_q$ endowed with uniform probability, each of which is associated with a parity-check code as follows: for each $H \in \R_{m,n}$, the corresponding parity-check code $C_H$ is given by
\begin{equation}\label{PC}
C_H=\{\mathbf{x} \in \F_q^n: H\mathbf{x}^T=\mathbf{0}\}.
\end{equation}
Here boldface letters such as $\mathbf{x}$ denote row vectors.

The ensemble $\R_{m,n}$ has been studied for a long time and many strong results have been obtained. For example, in the classical work of Gallager \cite{Gallager2}, an upper bound of the average number of codewords of a given weight in $\R_{m,n}$ was obtained, from which information about the minimum-distance distribution of $\R_{m,n}$ can be derived (see \cite[Theorems 2.1-2.2]{Gallager2}); in \cite{Berk} (see also \cite{Di}) union bounds on the block erasure probabilities of the ensemble $\R_{m,n}$ was obtained under the maximum likelihood decoding (see \cite{Berk,Di}). More recently, the undetected error probability of the ensemble $\R_{m,n}$ was studied in the binary symmetric channel by Wadayama \cite{Wadayama} (i.e. $q=2$), and some bounds on the error probability under the maximum likelihood decoding principle were obtained in the $q$-ary erasure channel \cite{Fashandi,Liva1}. It is easy to see that $\R_{m,n}$ contains all linear codes in the random $[n,n-m]_q$ code ensemble considered in \cite{FFW}, but these two ensembles are different for two reasons: first, in the random $[n,k]_q$ code ensemble considered in \cite{FFW}, each $[n,k]_q$ code is counted exactly once, while in $\R_{m,n}$ each code is counted with some multiplicity as different choices for the matrix $H$ may give rise to the same code; second, some codes in $\R_{m,n}$ may have rates strictly larger than $1-\frac{m}{n}$ as the rows of $H$ may not be linearly independent.

It is conceivable that most of the codes in $\R_{m,n}$ have rate $1-\frac{m}{n}$, and the average behavior of codes in $\R_{m,n}$ should be similar to that of the random $[n,n-m]_q$ code ensemble considered in \cite{FFW}. We will show that this is indeed the case. Actually we will obtain stronger results, taking advantage of the fine algebraic structure of the ensemble $\R_{m,n}$, which may not be readily available in the random $[n,k]_q$ code ensemble. Such structure has been exploited in \cite{Wadayama}.

We first obtain explicit formulas for the average decoding error probability of the ensemble $\R_{m,n}$ over the erasure channel under the three different decoding principles. This is comparable to \cite[Theorem 2]{FFW} for the random $[n,k]_q$ code ensemble. Such formulas are useful as they allow explicit evaluations of the average decoding error probabilities for any given $m$ and $n$, hence giving us a meaningful guidance as to what to expect for a good $[n,n-m]_q$ code over the erasure channel.

\begin{theorem}\label{Exp}
Let $\R_{m,n}$ be the random matrix ensemble described above. Denote by $\Gaubin{i}{j}_q$ the Gaussian $q$-binomial coefficient and denote
\begin{eqnarray} \label{1:phim} \psi_m(i):=\prod_{k=0}^{i-1}(1-q^{k-m}), \quad 1 \le i \le m.  \end{eqnarray}
\begin{enumerate}
\item The average unsuccessful decoding probability of $\R_{m,n}$ under list decoding with list size $q^\l$, where $\l$ is a non-negative integer, is given by
\begin{equation}\label{Pld}
P_{\mathrm{ld}}(\R_{m,n},\l,\ep)=\sum_{i=1}^n\sum_{j=\l+1}^iq^{-mj}\psi_m(i-j)\Gaubin{i}{j}_q\binom{n}{i}\ep^i(1-\ep)^{n-i};
\end{equation}
\item The average unsuccessful decoding probability of $\R_{m,n}$ under unambiguous decoding is given by
\begin{eqnarray} \label{Pud} P_{\mathrm{ud}}(\R_{m,n},\ep)=\sum_{i=1}^n\left( 1-\psi_m(i)\right)\binom{n}{i}\ep^i(1-\ep)^{n-i};\end{eqnarray}
\item The average decoding error probability of $\R_{m,n}$ under maximum likelihood decoding is given by
\begin{eqnarray} \label{Pmld} P_{\mathrm{mld}}(\R_{m,n},\ep)=\sum_{i=1}^n\sum_{j=1}^i q^{-mj}(1-q^{-j})\psi_m(i-j)\Gaubin{i}{j}_q\binom{n}{i}\ep^i(1-\ep)^{n-i}.\end{eqnarray}
\end{enumerate}
\end{theorem}

\begin{remark}
To make a comparison of Theorem 1 with corresponding results for the random $[n,k]_q$ code ensemble, we take $k=n-m$ and rewrite the explicit formulas obtained in \cite{FFW,Xio} as follows (see \cite[Theorem 2]{FFW} and \cite[Section 6]{Xio}): Denote by $\C_{m,n}$ the random $[n,n-m]_q$ code ensemble.
\begin{enumerate}
\item The average unsuccessful decoding probability of $\C_{m,n}$ under list decoding with list size $q^\l$, where $\l$ is a non-negative integer, is given by
\begin{equation}\label{oPld}
P_{\mathrm{ld}}(\C_{m,n},\l,\ep)=\sum_{i=1}^n\sum_{j=\l+1}^i\frac{\psi_{n-m}(j)}{\psi_n(i)}q^{-mj}\psi_m(i-j)\Gaubin{i}{j}_q\binom{n}{i}\ep^i(1-\ep)^{n-i};
\end{equation}
\item The average unsuccessful decoding probability of $\C_{m,n}$ under unambiguous decoding is given by
\begin{eqnarray} \label{oPud} P_{\mathrm{ud}}(\C_{m,n},\ep)=\sum_{i=1}^n\left( 1-\frac{\psi_m(i)}{\psi_n(i)}\right)\binom{n}{i}\ep^i(1-\ep)^{n-i};\end{eqnarray}
\item The average decoding error probability of $\C_{m,n}$ under maximum likelihood decoding is given by
\begin{eqnarray} \label{oPmld} P_{\mathrm{mld}}(\C_{m,n},\ep)=\sum_{i=1}^n\sum_{j=1}^i \frac{\psi_{n-m}(i)}{\psi_n(i)}q^{-mj}(1-q^{-j})\psi_m(i-j)\Gaubin{i}{j}_q\binom{n}{i}\ep^i(1-\ep)^{n-i}.\end{eqnarray}
\end{enumerate}
It is easy to check that for all indices $i,j$ appearing in (\ref{Pld}) and (\ref{oPld}) we have $0<\frac{\psi_{n-m}(j)}{\psi_n(i)}<1$, this shows that
\[P_{\mathrm{ld}}(\C_{m,n},\l,\ep)<P_{\mathrm{ld}}(\R_{m,n},\l,\ep).\]
Similarly we also have 
\[P_{\mathrm{ud}}(\C_{m,n},\ep)<P_{\mathrm{ud}}(\R_{m,n},\ep),\quad P_{\mathrm{mld}}(\C_{m,n},\ep)<P_{\mathrm{mld}}(\R_{m,n},\ep).\]
With this respect, the random $[n,k]_q$ code ensemble behaves slightly better than the $\R_{m,n}$ ensemble in terms of the average decoding error probabilities.
\end{remark}

Next, letting $m=(1-R)n$ for $0 < R < 1$, we compute the error exponents of the average decoding error probability of the ensemble series $\{\R_{(1-R)n,n}\}$ as $n \to \infty$ under these decoding principles.
\begin{theorem}\label{Errexp}
Let the rate $0 < R < 1$ be fixed and $n \to \infty$.
\begin{enumerate}
\item For any fixed integer $\l \ge 0$, the error exponent $T_{\mathrm{ld}}(\l,\ep)$ for average unsuccessful decoding probability of $\{\R_{(1-R)n,n}\}$ under list decoding with list size $q^\l$ is given by
\begin{equation}\label{Errld}
T_{\mathrm{ld}}(\l,\ep)=\begin{cases}
0 &(1-\ep \leq R < 1)\\
(1-R)\log_q\left(\frac{1-R}{\ep}\right)+R\log_q\left(\frac{R}{1-\ep}\right) &\left(\frac{1-\ep}{1-\ep+q^{\l+1}\ep} < R < 1-\ep\right)\\
(\l+1)(1-R)-\log_q\left(1-\ep+q^{\l+1}\ep\right) &\left(0 < R \leq \frac{1-\ep}{1-\ep+q^{\l+1}\ep}\right).
\end{cases}
\end{equation}
\item The error exponents for average unsuccessful decoding probability of $\{\R_{(1-R)n,n}\}$ under unambiguous decoding and maximum likelihood decoding (respectively) are both given by
\begin{equation}\label{Errud}
T_{\mathrm{ud}}(\ep)=T_{\mathrm{mld}}(\ep)=\begin{cases}
0 &(1-\ep \leq R < 1)\\
(1-R)\log_q\left(\frac{1-R}{\ep}\right)+R\log_q\left(\frac{R}{1-\ep}\right) &\left(\frac{1-\ep}{1-\ep+q \ep} < R < 1-\ep\right)\\
1-R-\log_q \left( 1-\ep+q \ep\right) &\left(0 < R \leq \frac{1-\ep}{1-\ep+q \ep}\right).
\end{cases}
\end{equation}
\end{enumerate}
\end{theorem}

A plot of the function $T_{\mathrm{ld}}(\l,\ep)$ for $q=2,\ep=0.25,\l=0,1,2$ in the range $0 < R < 1$ is given by {\bf Fig. 1}.

\begin{figure}
\begin{center}
\includegraphics[angle=0,width=0.8 \textwidth,height=0.35 \textheight]{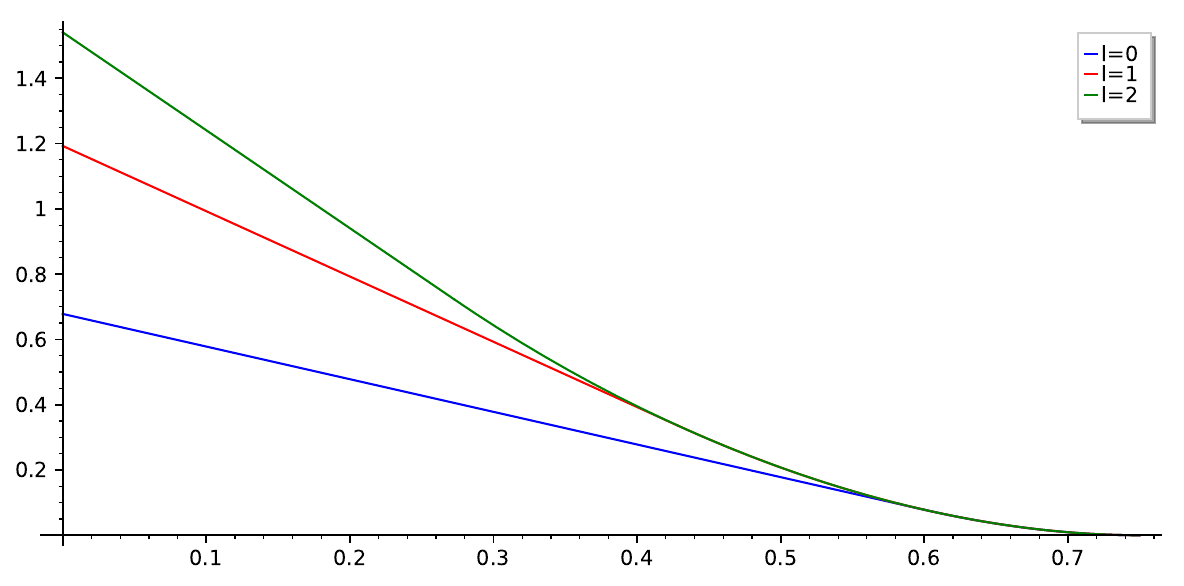}
\caption{The error exponent $T_{\mathrm{ld}}(\l,\ep)$ for $0<R<1$ where $\l=0,1,2$ and $q=2, \ep=0.25$.}
\end{center}
\end{figure}

It turns out that the error exponents obtained here under these decoding principles are identical with those for the random $[n,nR]_q$ code ensemble obtained in \cite[Theorem 3]{FFW} and \cite[Theorems 1.3 and 1.4]{Xio}. 

We establish a strong concentration result for the unsuccessful decoding probability of a random code in the ensemble $\R_{(1-R)n,n}$ towards the mean under unambiguous decoding.
\begin{theorem}\label{Concen} Let the rate $0 < R < 1$ be fixed and $n \to \infty$. Then as $H_n$ runs over the ensemble $\R_{(1-R)n,n}$, we have
\begin{eqnarray} \label{1:conc} \frac{P_\mathrm{ud}(H_n,\ep)}{P_\mathrm{ud}(\R_{(1-R)n,n},\ep)} \to 1 \quad \emph{\textbf{WHP}}, \end{eqnarray}
under either of the following conditions:
\begin{itemize}
\item[(1).] if $\frac{1-\ep}{1+(q-1)\ep^2} \leq R<1-\ep$ for any $q \ge 2$, or
\item[(2).] if $\frac{1-\ep}{1+(q-1)\ep} \le R<1-\ep$ for $q=2,3,4$.
\end{itemize}
\end{theorem}
Here the notion \textbf{WHP} in (\ref{1:conc}) refers to ``with high probability'', that is, for any $\delta>0$, there is $c>0$ and $n_0>0$ such that
\[\P \left(\left|\frac{P_\mathrm{ud}(H_n,\ep)}{P_\mathrm{ud}(\R_{(1-R)n,n},\ep)}-1\right|<\delta\right)
>1-q^{-nc} \quad \forall n >n_0.\]
Noting that in the range $0<R<1-\ep$, it was known that $T_\mathrm{ud}(\ep) > 0$ (see Theorem \ref{Errexp}), hence
$$P_\mathrm{ud}(\R_{(1-R)n,n},\ep)=q^{-n(T_\mathrm{ud}(\ep)+o(1))} \to 0, \quad \mbox{ as } n \to \infty, $$
so (\ref{1:conc}) shows that $P_\mathrm{ud}(H_n,\ep)$ also tends to zero exponentially fast with high probability for the ensemble $\R_{(1-R)n,n}$ under either Condition (1) or (2) of Theorem \ref{Concen}.

The paper is now organized as follows. In Section \ref{Pre}, we introduce the three decoding principles and the Gaussian $q$-binomial coefficient in more details. Then in Section \ref{countrank}, we provide three counting results regarding $m \times n$ matrices of certain rank over $\F_q$. Afterwards in Sections \ref{Pfexp}, \ref{Pferr} and \ref{Pfcon}, we give the proofs of Theorems \ref{Exp}, \ref{Errexp} and \ref{Concen} respectively. The proofs of Theorem \ref{Concen} involves some technical calculus computations on the error exponent of the variance. In order to streamline the proofs, we put some of the arguments in Section \ref{Appen1} \textbf{Appendix}. Finally we conclude this paper in Section \ref{Conclusion}.

\section{Preliminaries}\label{Pre}
\subsection{Three decoding principles, the average decoding error probability, and the error exponent}
The three decoding principles have been well studied in the literature. For the sake of readers, we explain these terms here, following the presentation outlined in \cite{FFW}.

In a $q$-ary erasure channel, the channel input alphabet is a finite field $\F_q$ of order $q$, and during transmission, each symbol $x \in \F_q$ is either received correctly with probability $1-\ep$ or erased with probability $\ep$ $(0<\ep<1)$.

Let $C$ be a code of length $n$ over $\F_q$. For a codeword ${\bf c}=(c_1,c_2,\ldots,c_n)\in C$ that was transmitted through the channel, suppose the word ${\bf r}=(r_1,r_2,\ldots,r_n)$ is received. Denote by $E$ the set of $i$'s such that $r_i=\square$, i.e., the $i$-th symbol was erased during the transmission. In this case, we say that an erasure set $E$ occurs. The probability $P(E)$ that this erasure set $E$ occurs for $\mathbf{c}$ is clearly $\ep^{\#E}(1-\ep)^{n-\#E}$. Here $\#E$ denote the cardinality of the set $E$.

Let $C(E,{\bf r})$ be the set of all codewords of $C$ that match the received word $\mathbf{r}$, that is,
\[C(E,\mathbf{r})=\{\mathbf{c} \in C: c_i=r_i \quad \forall i \in [n]\setminus E\}.\]
Here $[n]=\{1,2,\ldots,n\}$. The decoding problem is about how to choose one or more possible codewords in the set $C(E,\bf{r})$ when a word $\bf{r}$ is received. Now we consider three decoding principles:
\begin{itemize}
\item In {\bf unambiguous decoding}, the decoder outputs the only one codeword in $C(E,{\bf r})$ if $\#C(E,{\bf r})=1$ and declares ``failure'' otherwise. The unsuccessful decoding probability $P_{\mathrm{ud}}(C,\ep)$ of $C$ under unambiguous decoding is defined to be the probability that the decoder declares ``failure''. 
    
\item In {\bf list decoding}, the decoder with list size $q^{\l}$ outputs all the codewords in $C(E,\mathbf{r})$ if $\#C(E,\mathbf{r})\le q^{\l}$ and declares ``failure'' otherwise. The unsuccessful decoding probability $P_{\mathrm{ld}}(C,\l,\ep)$ of $C$ under list decoding with list size $q^{\l}$ is defined to be the probability that the decoder declares ``failure''.
    
\item In {\bf maximum likelihood decoding}, the decoder randomly choose a codeword in $C(E,\bf{r})$ uniformly and outputs this codeword. The decoding error probability $P_{\mathrm{mld}}(C,\ep)$ of $C$ under maximum likelihood decoding is defined to be the probability that the codeword outputted by the decoder is not the sent codeword. 
\end{itemize}
The computation of $P_{\mathrm{ud}}(C,\ep), P_{\mathrm{ld}}(C,\l,\ep)$ and $P_{\mathrm{mld}}(C,\ep)$ can be made much easier if $C$ is a linear code. 

Now assume that $C$ is an $[n,k]_q$ linear code, that is, $C$ is a $k$-dimensional subspace of $\F_q^n$. For any $E \subset [n]$, define
\[C(E):=\left\{\mathbf{c}=(c_1,\ldots,c_n) \in C: c_i=0 \,\, \forall i \in [n] \setminus E\right\}.\]
It is easy to see that if the set $C(E,\mathbf{r})$ is not empty, then the cardinality of $C(E,\mathbf{r})$ is the same as that of $C(E)$, which is a vector space over $\F_q$. Denote by $\{I_i^{(\l)}(C)\}_{i=1}^{n}$ the $q^{\l}$-incorrigible set distribution of $C$, and $\{I_i(C)\}_{i=1}^{n}$ the incorrigible set distribution of $C$, which are defined respectively as follows:
\begin{eqnarray} I_i^{(\ell)}(C)&=&\#\left\{E \subset [n]: \#E=i, \dim C(E) > \l\right\},\nonumber \\
\label{3:Ii0} I_i(C)&=&\#\left\{E \subset [n]: \#E=i, C(E) \ne \{0\}\right\}.\end{eqnarray}
We see that $\dim C(E) \le \min\{k, \# E\}$, so if $\l \ge \min\{k, \# E\}$, then $I_i^{(\ell)}(C)=\emptyset$. We also define
\begin{eqnarray} \label{pre:llambda} \lambda_i^{(\l)}(C)=\#\left\{E \subset [n]: \#E=i, \dim C(E) =\l\right\}.\end{eqnarray}
It is easy to see that $I_i^{(0)}(C)=I_i(C)$, $\lambda_i^{(\l)}(C)=0$ if $\l>\min\{i,k\}$ and
%\begin{eqnarray} \label{2:biid} \lambda_i^{(\l)}:=I_i^{(\ell-1)}(C)-I_i^{(\ell)}(C),\end{eqnarray}
%and
\begin{eqnarray} \label{3:lamsum} \sum_{\l=0}^{i}\lambda_i^{(\l)}(C)=\binom{n}{i}.  \end{eqnarray}
We also have the identity
\begin{eqnarray} \label{3:il}
I^{(\l)}_i(C)=\sum_{j=\l+1}^{i}\lambda_i^{(j)}(C), \quad \forall i, \l.
\end{eqnarray}
Recall from \cite{FFW} that the values $P_{\mathrm{ud}}(C,\ep), P_{\mathrm{ld}}(C,\ep)$ and $P_{\mathrm{ld}}(C,\ell,\ep)$ can all be expressed in terms of $I_i^{(\ell)}(C)$, $I_i(C)$ and $\lambda_i^{(\l)}(C)$ as follows:
\begin{eqnarray*}
%\label{3:pud}
P_{\mathrm{ud}}(C,\ep)&=&\sum_{i=1}^n I_i(C) \ep^i (1-\ep)^{n-i}, \\
%\label{3:pld}
P_{\mathrm{ld}}(C,\ell,\ep)&=&\sum_{i=1}^n I_i^{(\l)}(C) \ep^i (1-\ep)^{n-i}, \\
%\label{3:pmld}
P_{\mathrm{mld}}(C,\ep)&=&\sum_{i=1}^n \sum_{\l=1}^i \lambda_i^{(\l)}(C) (1-q^{-\l})\ep^i (1-\ep)^{n-i}.
\end{eqnarray*}

For $H \in \R_{m,n}$, we write $P_\mathrm{ld}(H,\l,\ep):=P_\mathrm{ld}(C_H,\l,\ep)$ and $P_*(H,\ep):=P_*(C_H,\ep)$ for $* \in \{\mathrm{ud}, \mathrm{mld}\}$, where $C_H$ is the parity-check code defined by (\ref{PC}). The \emph{average decoding error probabilities} over the ensemble $\R_{m,n}$ are given by
$$P_\mathrm{ld}(\R_{m,n},\l,\ep):=\E[P_\mathrm{ld}(H,\l,\ep)],$$
and
$$P_*(\R_{m,n},\ep):=\E[P_*(H,\ep)], \quad * \in \{\mathrm{ud}, \mathrm{mld}\}.$$
Here the expectation $\E$ is taken over the ensemble  $\R_{m,n}$.

Let $0<R, \ep<1$ be fixed constants. The error exponent $T_{\mathrm{ud}}(\ep)$ for average unsuccessful decoding probability of the family of ensembles $\{\R_{(1-R)n,n}\}$ is defined as
\[T_{\mathrm{ud}}(\ep):=-\lim_{n \to \infty} \frac{1}{n} \log_q P_{\mathrm{ud}}(\R_{(1-R)n,n},\ep), \]
provided that the limit exists \cite{Gallager,FFW,Viterbi}. The error exponents of $\{\R_{(1-R)n,n}\}$ for the other two decoding principles are defined similarly.

\subsection{Gaussian binomial coefficients}
For integers $n \ge k \ge 0$, the Gaussian binomial coefficients $\Gaubin{n}{k}_q$ is defined as
\[\Gaubin{n}{k}_q:=\frac{(q)_n}{(q)_k (q)_{n-k}},\]
where $(q)_n:=\prod_{i=1}^n \left(1-q^i\right)$. By convention $(q)_0=1$, $\Gaubin{n}{0}_q=1$ for any $n \ge 0$ and $\Gaubin{n}{k}_q=0$ if $k<0$ or $k>n$. The function $\psi_m(i)$ defined in (\ref{1:phim}) can be written as
\[\psi_m(i)=q^{-mi+i(i-1)/2}(q)_i\Gaubin{m}{i}_q. \]
We may define $\psi_m(0)=1$ for $m \ge 0$ and $\psi_m(i)=0$ if $i<0$ or $i>m$. Next, recall the well-known combinatorial interpretation of $\Gaubin{n}{k}_q$:
\begin{lemma}[\cite{MacWilliams}] \label{2:lem1} The number of $k$-dimensional subspaces of an $n$-dimensional vector space over $\F_q$ is $\Gaubin{n}{k}_q$.
\end{lemma}
The Gaussian binomial coefficient $\Gaubin{n}{k}_q$ satisfies the property
\[\Gaubin{n}{k}_q=\Gaubin{n}{n-k}_q, \qquad \forall \,\,\, n \ge k \ge 0,\]
and the identity
\begin{eqnarray} \label{2:gid}
\prod_{i=0}^{n-1} \left(1+q^ix\right)=\sum_{i=0}^n \Gaubin{n}{i}_q q^{i(i-1)/2}  x^i, \quad \forall \,\, n \ge 1.  \end{eqnarray}

\section{Three counting results for the ensemble $\R_{m,n}$}\label{countrank}

In this section we provide three counting results about matrices of certain rank in the ensemble $\R_{m,n}$. Such results may not be new, but since it is not easy to locate them in the literature, we prove them here. These results will be used repeatedly in the proofs later on.

For $H \in \R_{m,n}$, denote by $\rank(H)$ the rank of the matrix $H$ over $\F_q$.

\begin{lemma}\label{rank}
Let $H$ be a random matrix in the ensemble $\R_{m,n}$. Then for any integer $j$, we have
\begin{eqnarray} \label{2:rk1id} \P(\rank(H)=j)=q^{-m(n-j)}\psi_m(j)\Gaubin{n}{j}_q.\end{eqnarray}
\end{lemma}
\begin{proof}
We may assume that $j$ satisfies $0 \le j \le n$, because if $j$ is not in the range, then both sides of Equation (\ref{2:rk1id}) are obviously zero.

Denote by $\mathrm{Hom}(m,n)$ the set of $\F_q$-linear transformations from $\F_q^m$ to $\F_q^n$. Writing vectors in $\F_q^m$ and $\F_q^n$ as row vectors, we see that the random matrix ensemble $\R_{m,n}$ can be identified with the set $\mathrm{Hom}(m,n)$ via the relation
\begin{equation} \label{2:iden1} H \,\,\, \leftrightarrow \,\,\, G : \begin{array}{c}\F_q^m \to \F_q^n\\
\mathbf{x} \mapsto \mathbf{x}H
\end{array}.\end{equation}

Since $\rank(H)=j$ if and only if $\dim (\mathrm{Im}G)=j$, and $\#\R_{m,n}=q^{mn}$, we have
\begin{align*}
\P(\rank(H)=j)&=q^{-mn}\sum_{\substack{G \in \mathrm{Hom}(m,n) \\ \dim(\mathrm{Im}G)=j}} 1\\
&=q^{-mn}\sum_{\substack{V \leq \F_q^n \\ \dim V=j}}\sum_{\substack{G \in \mathrm{Hom}(m,n) \\ \mathrm{Im}G=V}} 1\, .
\end{align*}
The inner sum $\sum_{G \in \mathrm{Hom}(m,n),\mathrm{Im}G=V} 1$ counts the number of surjective linear transformations from $\F_q^m$ to $V$, a $j$-dimensional subspace of $\F_q^n$. Since $V \cong \F_q^j$, this is also the number of surjective linear transformations from $\F_q^m$ to $\F_q^j$, or, equivalently, the number of $m \times j$ matrices $K$ over $\F_q$ such that the columns of $K$ are linearly independent. The number of such matrices $K$ can be counted as follows: the first column of $K$ can be any nonzero vector over $\F_q$, there are $q^m-1$ choices; given the first column, the second column can be any vector lying outside the space of scalar multiples of the first column, so there are $q^m-q$ choices; inductively, given the first $k$ columns, the $(k+1)$-th column lies outside a $k$-dimensional subspace, so the number of choices for the $(k+1)$-th column is $q^m-q^k$. Thus we have
\begin{eqnarray} \label{2:iden2}
\sum_{\substack{G \in \mathrm{Hom}_q(m,n) \\ \mathrm{Im}G=V}} 1
=\prod_{k=0}^{j-1} (q^m-q^k)
=q^{mj}\psi_m(j), \quad \dim V=j.
\end{eqnarray}
Together with Lemma \ref{2:lem1}, we obtain
$$\P(\rank(H)=j)=q^{-mn}q^{mj}\psi_m(j)\sum_{\substack{V \leq \F_q^n \\ \dim V=j}} 1=q^{-m(n-j)}\psi_m(j)\Gaubin{n}{j}_q,$$
which is the desired result.
\end{proof}

\begin{lemma}\label{rank2}
Let $H$ be a random matrix in the ensemble $\R_{m,n}$. Let $A \subset [n]:=\{1,2,\cdots,n\}$ be a subset with cardinality $s$. Denote by $H_A$ the $m \times s$ submatrix formed by columns of $H$ indexed from $A$. Then for any integers $j$ and $r$, we have
\begin{eqnarray} \label{2:rk2id} \P(\rank(H)=j \cap \rank(H_A)=r)=q^{-m(n-j)+r(n-j-s+r)}\psi_m(j)\Gaubin{s}{r}_q\Gaubin{n-s}{j-r}_q.\end{eqnarray}
\end{lemma}
\begin{proof}

We may assume that $0 \le j \le n$ and $\max\{0,s-n+j\} \le r \le \min\{j,s\}$, because if $j$ or $r$ does not satisfy this condition, then both sides of Equation (\ref{2:rk2id}) are zero.

Using the relations (\ref{2:iden1}) and (\ref{2:iden2}), we can expand the term $\P(\rank(H)=j \cap \rank(H_A)=r)$ as
\begin{equation}\label{prorank2}
\P(\rank(H)=j \cap \rank(H_A)=r)
=q^{-mn}\sum_{\substack{V \leq \F_q^n\\
\dim V=j \\
\dim V_A=r}} \sum_{\substack{G \in \mathrm{Hom}(m,n)\\ \mathrm{Im} G=V}} 1=q^{-m(n-j)}\psi_m(j)\sum_{\substack{V \leq \F_q^n\\ \dim V=j \\ \dim V_A=r}} 1.
\end{equation}
Here $V_A$ is the subspace of $\F_q^A$ formed by restricting $V$ to coordinates with indices from $A$. We may consider the projection given by
\begin{align*}
\pi_A: & V \to V_A\\
&(v_k)_{k=1}^n \mapsto (v_k)_{k \in A}\,.
\end{align*}
The kernel of $\pi_A$ has dimension $j-r$ and is of the form $W \times \{(0)_A\}$ for some subspace $W \le \F_q^{[n]-A}$. So we can further decompose the sum on the right hand side of (\ref{prorank2}) as
\begin{equation}\label{prorank3}
\sum_{\substack{V \leq \F_q^n\\
\dim V=j \\
\dim V_A=r}} 1=\sum_{\substack{W \leq \F_q^{[n]-A}\\
\dim W=j-r}} \sum_{\substack{V \leq \F_q^n \\
\dim V=j \\ \ker(\pi_A)=W \times \{(0)_A\} }} 1.
\end{equation}
Now we compute the inner sum on the right hand side of (\ref{prorank3}). Suppose we are given an ordered basis of the $(j-r)$-dimensional subspace $W \times \{(0)_A\}$ of $\F_q^n$. We extend it to an ordered basis of some $j$-dimensional subspace $V$ as follows: first we need $r$ other basis vectors $\mathbf{v}_1,\mathbf{v}_2,\cdots,\mathbf{v}_r$ to be linearly independent. At the same time, they have to be linearly independent with any nonzero vector in $\F_q^{[n]-A} \times \{(0)_A\}$ due to the kernel condition. This requires the set $\{\pi_A(\mathbf{v}_1), \pi_A(\mathbf{v}_2), \cdots, \pi_A(\mathbf{v}_r)\}$ to be linearly independent in $\F_q^A$. On the other hand, if this condition is satisfied, then the $r$ vectors $\mathbf{v}_1,\mathbf{v}_2,\cdots,\mathbf{v}_r$ are also linearly independent with one another as well as with any nonzero vector in $W \times \{(0)_A\}$. Therefore it reduces to counting the number of ordered linearly independent sets of $r$ vectors in $\F_q^A$. This number is clearly given by $\prod_{k=0}^{r-1}(q^s-q^k)$, so the total number of different ordered bases is given by $q^{r(n-s)}\prod_{k=0}^{r-1}(q^s-q^k)$.

On the other hand, given a fixed $j$-dimensional subspace $V$ with $\ker(\pi_A)=W \times \{(0)_A\}$, we count the number of ordered bases of $V$ of the form stated in previous paragraph as follows: we choose $\mathbf{v}_1$ to be any vector in $V$ but not in $W \times \{(0)_A\}$, which gives $q^j-q^{j-r}$ many choices for $\mathbf{v}_1$; similarly $\mathbf{v}_2$ is any vector in $V$ but not in the span of $W \times \{(0)_A\}$ and $\mathbf{v}_1$, this gives us $q^j-q^{j-r+1}$ many choices for $\mathbf{v}_2$; using this argument, we see that the number of such ordered bases is given by $\prod_{k=0}^{r-1} (q^j-q^{j-r+k})$.

We conclude from the above arguments that
\begin{align*}
\sum_{\substack{V \leq \F_q^n \\
\dim V=j \\
\ker(\pi_A)=W \times \{(0)_A\} }} 1&=\frac{q^{r(n-s)}\prod_{k=0}^{r-1}(q^s-q^k)}{\prod_{k=0}^{r-1} (q^j-q^{j-r+k})}\\
%&=\frac{q^{r(n-s)}\prod_{k=0}^{r-1}(1-q^{s-k})}{q^{r(j-r)}\prod_{k=0}^{r-1} (1-q^k)}\\
&=q^{r(n-j-s+r)}\Gaubin{s}{r}_q.
\end{align*}
Putting this into the right hand side of (\ref{prorank3}) and using Lemma \ref{2:lem1} again, we get
$$\sum_{\substack{V \leq \F_q^n\\ \dim V=j \\ \dim V_A=r}} 1=\sum_{\substack{W \leq \F_q^{[n]-A}\\ \dim W=j-r}} q^{r(n-j-s+r)}\Gaubin{s}{r}_q=q^{r(n-j-s+r)}\Gaubin{s}{r}_q\Gaubin{n-s}{j-r}_q.$$
The desired result is obtained immediately by plugging this into the right hand side of (\ref{prorank2}). This completes the proof of Lemma \ref{rank2}.
\end{proof}

\begin{lemma} \label{rank3} Let $H$ be a random matrix in the ensemble $\R_{m,n}$. Let $E,E'$ be subsets of $[n]$ such that
\[i=\#E, \quad
i'=\#E', \quad s=\#(E \cap E').\]
Then
\begin{eqnarray} \label{5:rank3}
\P(\rank(H_E)=i \cap \rank(H_{E'})=i')= \frac{\psi_m(i)\psi_m(i')}{\psi_m(s)}.
\end{eqnarray}
\end{lemma}

\begin{proof}
It is clear that if a matrix $H$ has full rank, then so is the submatrix $H_A$ for any index subset $A$. Hence we have
\begin{align*}
&\P(\rank(H_E)=i \cap \rank(H_{E'})=i')\\
&=\P(\rank(H_E)=i \cap \rank(H_{E'})=i' | \rank(H_{E \cap E'})=s)\P(\rank(H_{E \cap E'})=s).
\end{align*}
It is easy to see that the two events $\rank(H_E)=i$ and $\rank(H_{E'})=i'$ are conditionally independent given $\rank(H_{E \cap E'})=s$, since columns of $H_E$ and $H_{E'}$ are independent as random vectors over $\F_q$. Hence we get
\begin{align*}
&\P(\rank(H_E)=i \cap \rank(H_{E'})=i')\\
&=\P(\rank(H_E)=i | \rank(H_{E \cap E'})=s)\P(\rank(H_E)=i' | \rank(H_{E \cap E'})=s)\P(\rank(H_{E \cap E'})=s)\\
&=\frac{\P(\rank(H_E)=i \cap \rank(H_{E \cap E'})=s)\P(\rank(H_E)=i' \cap  \rank(H_{E \cap E'})=s)}{\P(\rank(H_{E \cap E'})=s)}\\
&=\frac{\psi_m(i)\psi_m(i')}{\psi_m(s)}.
\end{align*}
Here we have applied Lemmas \ref{rank} and \ref{rank2} in the last equality with $j=i, j'=i'$ and $r=s, H=H_E$ and $A=E \cap E'$.
\end{proof}

\section{Proof of Theorem \ref{Exp}}\label{Pfexp}

%\subsection{Three decoding principles, decoding error probabilities, and the error exponent}

%\subsection{Three decoding principles, decoding error probabilities, and the error exponent}

Let $C$ be an $[n,k]_q$ linear code. Recall from Section \ref{Pre} that the values $P_{\mathrm{ud}}(C,\ep), P_{\mathrm{ld}}(C,\ep)$ and $P_{\mathrm{ld}}(C,\ell,\ep)$ can all be expressed explicitly as
\begin{eqnarray}
\label{3:pud}
P_{\mathrm{ud}}(C,\ep)&=&\sum_{i=1}^n I_i(C) \ep^i (1-\ep)^{n-i}, \\
\label{3:pld}
P_{\mathrm{ld}}(C,\ell,\ep)&=&\sum_{i=1}^n I_i^{(\l)}(C) \ep^i (1-\ep)^{n-i}, \\
\label{3:pmld}
P_{\mathrm{mld}}(C,\ep)&=&\sum_{i=1}^n \sum_{\l=1}^i \lambda_i^{(\l)}(C) (1-q^{-\l})\ep^i (1-\ep)^{n-i},
\end{eqnarray}
where $I_i(C)$ and $I_i^{(\ell)}(C)$ are the incorrigible set distribution of $C$ and the $q^{\l}$-incorrigible set distribution of $C$ respectively, and $\lambda_i^{(\l)}(C)$ is defined in (\ref{pre:llambda}) also in Section \ref{Pre}.

Now we can start the proof of Theorem \ref{Exp}. For $H \in \R_{m,n}$, we denote $$I_i:=I_i(C_H),\,I_i^{(\l)}:=I_i^{(\l)}(C_H),\,\lambda_i^{(\l)}:=\lambda_i^{(\l)}(C_H).$$

%\subsection{Proof of Theorem \ref{Exp}}

Taking expectations on both sides of Equations (\ref{3:pud})-(\ref{3:pmld}) over the ensemble $\R_{m,n}$, we obtain
\begin{eqnarray}
\label{3:pud2}
P_{\mathrm{ud}}(\R_{m,n},\ep)&=&\sum_{i=1}^n \E[I_i] \ep^i (1-\ep)^{n-i}, \\
\label{3:pld2}
P_{\mathrm{ld}}(\R_{m,n},\ell,\ep)&=&\sum_{i=1}^n \E[I_i^{(\l)}] \ep^i (1-\ep)^{n-i}, \\
\label{3:pmld2}
P_{\mathrm{mld}}(\R_{m,n},\ep)&=&\sum_{i=1}^n \sum_{\ell=1}^i \E[\lambda_i^{(\l)}] (1-q^{-\l})\ep^i (1-\ep)^{n-i}.
\end{eqnarray}
We now compute $\E[\lambda_i^{(\l)}]$. Noting that for $H \in \R_{m,n}$ and $E \subset [n]$, we have $\rank(H_E)+\dim C_H(E)=\#E$, thus
\begin{eqnarray*} \E[\lambda_i^{(\l)}]&=&\frac{1}{\#\R_{m,n}} \sum_{\substack{H \in \R_{m,n}}} \sum_{\substack{E \subset [n]\\
\#E=i}}  \mathbbm{1}_{\dim C_H(E)=\l}\\
&=&\frac{1}{\#\R_{m,n}} \sum_{\substack{E \subset [n]\\
\#E=i}} \sum_{\substack{H \in \R_{m,n}\\
\rank(H_E)=i-\l}} 1.
\end{eqnarray*}
By the symmetry of the ensemble $\R_{m,n}$, the inner sum on the right hand side depends only on the cardinality of $E$, so we may assume $E=[i]$ to obtain
\begin{eqnarray*} \E[\lambda_i^{(\l)}]=\frac{1}{\#\R_{m,i}} \binom{n}{i}\sum_{\substack{H \in \R_{m,i}\\
\rank(H)=i-\l}}1.
\end{eqnarray*}
The right hand side is exactly $\binom{n}{i}\P(\rank(H)=i-\l)$ where the probability is over the ensemble $\R_{m,i}$. So from Lemma \ref{rank} we have
\begin{eqnarray} \label{3:lamb} \E[\lambda_i^{(\l)}]=q^{-m\l} \psi_m(i-\l) \Gaubin{i}{i-\l}_q\binom{n}{i}.
\end{eqnarray}
Using this and (\ref{3:il}), we also obtain
\begin{eqnarray*} \label{3:Iil} \E[I^{(\l)}_i]=\sum_{j=\l+1}^i \E[\lambda_i^{(j)}]=\sum_{j=\l+1}^iq^{-mj} \psi_m(i-j) \Gaubin{i}{j}_q \binom{n}{i}.\end{eqnarray*}
Inserting the above values $\E[I_i^{(\l)}]$ and $\E[\lambda_i^{(\l)}]$ into (\ref{3:pld2}) and (\ref{3:pmld2}) respectively, we obtain explicit expressions of $P_{\mathrm{ld}}(\R_{m,n},\l,\ep)$ and $P_{\mathrm{mld}}(\R_{m,n},\ep)$, which agree with (\ref{Pld}) and (\ref{Pmld}) of Theorem \ref{Exp}.

As for $P_{\mathrm{ud}}(\R_{m,n},\ep)$, noting by (\ref{3:lamsum}) that
$\sum_{\l=0}^i \E[\lambda_i^{(\l)}]=\binom{n}{i}$ and by (\ref{3:lamb}) that $\E[\lambda_i^{(0)}]=\psi_m(i)\binom{n}{i}$, therefore
\begin{eqnarray} \label{3:eli} \E[I_i]=\sum_{j=1}^i \E[\lambda_i^{(j)}]=\binom{n}{i}-\E[\lambda_i^{(0)}]=\left(1-\psi_m(i)\right)\binom{n}{i}. \end{eqnarray}
Inserting this value into (\ref{3:pud2}), we also obtain the desired expression of $P_{\mathrm{ud}}(\R_{m,n},\ep)$. This completes the proof of Theorem \ref{Exp}.

\section{Proof of Theorem \ref{Errexp}}\label{Pferr}

First recall that the error exponents of the average decoding error probability of the ensemble $\R_{(1-R)n,n}$ over the erasure channel under the three decoding principles are defined by
\begin{eqnarray} \label{5:tld} T_{\mathrm{ld}}(\l,\ep):=-\lim_{n \to \infty}\frac{1}{n}\log_q P_{\mathrm{ld}}(\R_{(1-R)n,n},\l,\ep),\end{eqnarray}
and
$$T_*(\ep):=-\lim_{n \to \infty}\frac{1}{n}\log_q P_*(\R_{(1-R)n,n},\ep), \quad * \in \{\mathrm{ud}, \mathrm{mld}\},$$
provided that the limit exists \cite{Gallager,FFW,Viterbi}.

Unambiguous decoding corresponds to list decoding with $\l=0$, and it is also easy to see that \[\frac{1}{2}P_{\mathrm{ud}}(C,\ep) \le P_{\mathrm{mld}}(C,\ep) \le P_{\mathrm{ud}}(C,\ep), \]
hence we have
\[T_{\mathrm{ud}}(\ep)=T_{\mathrm{mld}}(\ep)=T_{\mathrm{ld}}(0,\ep),\]
if the limit exists say in (\ref{5:tld}) for $\l=0$. So we only need to prove Part 1) of Theorem \ref{Errexp} for the case of list decoding.

Write $m=(1-R)n$, and we define
$$f_{i,j}:=q^{-mj}\psi_m(i-j)\Gaubin{i}{j}_q\binom{n}{i}\ep^i(1-\ep)^{n-i}.$$
It is easy to see that $f_{i,j} \geq 0$ for all integers $i,j$. In addition, $f_{i,j} \neq 0$ if and only if $1 \leq i \leq n, 1 \leq j \leq i$ and $i-j \leq m$.

We can rewrite (\ref{Pld}) as
$$P_{\mathrm{ld}}(\R_{m,n},\l,\ep)=\sum_{i,j} f_{i,j},$$
where $i,j$ are integers satisfying the conditions
\begin{equation}\label{conditionld}
\l+1 \leq i \leq n, \quad \max\{i-m,\l+1\} \leq j \leq i.
\end{equation}
The number of such integer pairs $(i,j)$ is at most $mn$. Noting that $\lim_{n \to \infty} \frac{\log_q(mn)}{n}=0$, we have
\begin{equation}\label{erreq}
\frac{1}{n}\log_q P_{\mathrm{ld}}(\R_{m,n},\l,\ep)=o(1)+\frac{1}{n}\max_{i,j}\log_q f_{i,j}.
\end{equation}
Now we focus on the quantity $f_{i,j}$. First, set $T=q^{-1}$ so that $0 < T < 1$. It is easy to verify that for $0 \leq j \leq i$,
\begin{eqnarray} \label{4:qtot}\Gaubin{i}{j}_q=q^{j(i-j)}\Gaubin{i}{j}_T, \quad \psi_m(i)=\frac{(T)_m}{(T)_{m-i}}.\end{eqnarray}
Hence we have
$$f_{i,j}=q^{-j(m-i+j)}\frac{(T)_m}{(T)_{m-i+j}}\Gaubin{i}{j}_T\binom{n}{i}\ep^i(1-\ep)^{n-i}.$$
The infinite product $(T)_\infty:=\prod_{r=1}^\infty(1-T^r)$ converges absolutely to some positive real number $M$ which only depends on $q$, and $M=(T)_\infty < (T)_m \leq 1$ for any $m$. This implies that
$$M^2 \leq \frac{(T)_m}{(T)_{m-i+j}}\Gaubin{i}{j}_T =\frac{(T)_m(T)_i}{(T)_{m-i+j}(T)_j(T)_{i-j}} \leq \frac{1}{M^3}, $$
and thus
$$\frac{1}{n}\log_q\left(\frac{(T)_m}{(T)_{m-i+j}}\Gaubin{i}{j}_T\right)=o(1).$$
Therefore we have
$$\frac{1}{n}\log_q f_{i,j}=-\frac{j(m-i+j)}{n}+\frac{1}{n}\log_q \binom{n}{i}+\frac{i}{n}\log_q \ep+\frac{n-i}{n}\log_q(1-\ep)+o(1).$$
We want to maximize this quantity over $i,j$ satisfying (\ref{conditionld}). Since the term $-\frac{j(m-i+j)}{n}$ is always non-positive, it is easy to see that for any fixed $i$, to maximize the term $\frac{1}{n}\log_q f_{i,j}$, we shall take
\[j=\left\{
\begin{array}{lcl}
\l+1 &:& \mbox{ if } \l+1 > i-m, \mbox{ or equivalently if } i < m+\l+1, \\
i-m &:& \mbox{ otherwise}.
\end{array}\right.\]
So we can simplify (\ref{erreq}) as
\begin{align*}
&\quad\frac{1}{n} \log_q P_{\mathrm{ld}}(\R_{m,n},\l,\ep)\\
&=o(1)+\max_{\l+1 \leq i \leq n}\left[-\frac{(\l+1)(m-i+\l+1)_+}{n}+\frac{1}{n}\log_q \binom{n}{i}+\frac{i}{n}\log_q \ep+\frac{n-i}{n}\log_q(1-\ep)\right],
\end{align*}
where $(x)_+:=\max\{0,x\}$.

Let $i=tn$ with $0 < t \leq 1$. Using the result (see \cite{MacWilliams} and \cite[Lemma 3]{FFW} for example)
\begin{equation}\label{entropy}
\frac{1}{n}\log_q \binom{n}{tn}=h(t)+o(1),
\end{equation}
where $h(t)=-t\log_q t-(1-t)\log_q(1-t)$ is the binary entropy function (in $q$-its), and taking $n \to \infty$, we obtain
$$T_{\mathrm{ld}}(\l,\ep)=-\sup_{0 < t \leq 1} f(t),$$
where
\begin{equation}\label{f}
f(t)=-(\l+1)(1-R-t)_+ +h(t)+t\log_q\ep+(1-t)\log_q(1-\ep).
\end{equation}
Note that $f(t)$ is continuous on the interval $(0,1)$.

Differentiating (\ref{f}) with respect to $t$, we get
$$\frac{\d f}{\d t}=\begin{cases}
\l+1+\log_q\left(\frac{\ep}{t}\right)-\log_q\left(\frac{1-\ep}{1-t}\right) &(0 < t < 1-R)\\
\log_q\left(\frac{\ep}{t}\right)-\log_q\left(\frac{1-\ep}{1-t}\right) &(1-R < t < 1)
\end{cases}$$
It is easy to check from the right hand side that $\frac{\d f}{\d t}=0$ at \[t_1=\frac{q^{\l+1}\ep}{1+(q^{\l+1}-1)\ep}, \quad t_2=\ep, \]
and at these points the function $f(t)$ has a local maximum. There are three cases to consider:

\textbf{Case 1:} $0 < R \leq \frac{1-\ep}{1+(q^{\l+1}-1)\ep}$

In this case we have $\ep < t_1 \leq 1-R$, then $f(t)$ is maximized at $t=t_1$, so
$$T_{\mathrm{ld}}(\l,\ep)=-f(t_1)=(\l+1)(1-R)-\log_q[1+(q^{\l+1}-1)\ep].$$

\textbf{Case 2:} $\frac{1-\ep}{1+(q^{\l+1}-1)\ep} < R < 1-\ep$

In this case we have $\ep < 1-R < t_1$, then $f(t)$ is maximized at $t=1-R$, so
$$T_{\mathrm{ld}}(\l,\ep)=-f(1-R)=(1-R)\log_q\left(\frac{1-R}{\ep}\right)+R\log_q\left(\frac{R}{1-\ep}\right).$$

\textbf{Case 3:} $1-\ep \leq R < 1$

In this case we get $1-R \leq \ep < t_1$, then $f(t)$ is maximized at $t=t_2=\ep$, so
$$T_{\mathrm{ld}}(\l,\ep)=-f(\ep)=0.$$
Combining {\bf Cases 1-3} above gives Equation (\ref{Errld}). This completes the proof of Theorem \ref{Errexp}.

\section{Proof of Theorem \ref{Concen}}\label{Pfcon}
The proof of Theorem \ref{Concen} depend on the computation of the variance of the unsuccessful decoding probability under unambiguous decoding and its error exponent.

\subsection{The variance of unsuccessful decoding probability and its error exponent}
Note from (\ref{3:pud}) that the variance $\sigma_{\mathrm{ud}}^2(\R_{m,n},\ep)$ of the unsuccessful decoding probability under unambiguous decoding can be expressed as
\begin{align*}
\sigma_{\mathrm{ud}}^2(\R_{m,n},\ep)&=\E[P_\mathrm{ud}(H,\ep)^2]-\left(\E[P_\mathrm{ud}(H,\ep)]\right)^2\\
&=\E\left[\left(\sum_{i=1}^n I_i \ep^i(1-\ep)^{n-i}\right)^2\right]-\left(\sum_{i=1}^n \E[I_i]\ep^i(1-\ep)^{n-i}\right)^2\\
&=\sum_{i,i'=1}^n\Cov(I_i,I_{i'})\ep^{i+i'}(1-\ep)^{2n-i-i'},
\end{align*}
where the term $\Cov(I_i,I_{i'})$ is given by
\[\Cov(I_i,I_{i'})=\E[I_iI_{i'}]-\E[I_i]\E[I_{i'}]. \]
We first obtain:
\begin{lemma}\label{Cov}
For $1 \leq i,i' \leq n$, we have
\begin{eqnarray} \label{5:covv} \Cov(I_i,I_{i'})=\psi_m(i)\psi_m(i')\sum_{s=1}^{\min\{i,i'\}}
\left(\frac{1}{\psi_m(s)}-1\right)\binom{n}{s, i-s, i'-s, n-i-i'+s}.\end{eqnarray}
Here the multinomial coefficient $\binom{n}{a,b,c,d}$ for any non-negative integers $a,b,c,d,n$ such that $a+b+c+d=n$ is given by
\[\binom{n}{a,b,c,d}:=\frac{n!}{a!b!c!d!}\,. \]
%When $m < \max\{i,i'\} \leq n$, $\Cov(I_i,I_{i'})=0$.
\end{lemma}
\begin{proof}[Proof of Lemma \ref{Cov}]
Obviously if $i$ or $i'$ does not satisfy the relation $1 \le i, i' \le m$, then both sides of (\ref{5:covv}) are zero. So we may assume that $1 \le i, i' \le m$.

For any $H\in \R_{m,n}$, from (\ref{3:Ii0}) we see that
\begin{eqnarray*} I_i(C_H)
=\sum_{\substack{E \subset [n]\\
\#E=i}} \left( 1-\mathbbm{1}_{\dim C_H(E)=0}\right)
=\sum_{\substack{E \subset [n]\\
\#E=i}} \left( 1-\mathbbm{1}_{\rank (H_E)=i}\right)\,.\end{eqnarray*}
So we can expand the term $\Cov(I_i,I_{i'})$ as
\begin{align}
& \Cov(I_i,I_{i'})\nonumber\\
&=\E[I_iI_{i'}]-\E[I_i]\E[I_{i'}]\nonumber\\
&=\sum_{\substack{E \subset [n]\\ \# E=i}} \sum_{\substack{E' \subset [n]\\ \# E'=i'}} \E[\mathbbm{1}_{\rank (H_E)=i}\mathbbm{1}_{\rank (H_E')=i'}]-\E[\mathbbm{1}_{\rank (H_E)=i}]\E[\mathbbm{1}_{\rank (H_E')=i'}] \nonumber\\
&=\sum_{\substack{E \subset [n]\\ \# E=i}} \sum_{\substack{E' \subset [n]\\ \# E'=i'}}
\Big[\P(\rank(H_E) = i \cap \rank(H_{E'}) = i')-\P(\rank(H_E) = i)\P(\rank(H_{E'}) = i')\Big]. \label{Coveq}
\end{align}
Applying Lemmas \ref{rank} and \ref{rank3} to the terms in the inner sums on the right hand side of (\ref{Coveq}), we have
\begin{equation}\label{Coveq2}
\Cov(I_i,I_{i'})=\sum_{s}\sum_{\substack{E,E' \subset [n]\\ \# E=i, \#E'=i'\\
\#(E \cap E')=s}} \left(\frac{\psi_m(i)\psi_m(i')}{\psi_m(s)}-\psi_m(i)\psi_m(i')\right).
\end{equation}
Together with the combinatorial identity
\begin{equation}\label{Comb}
\sum_{\substack{E,E' \subset [n]\\ \# E=i,\# E'=i'\\ \#(E \cap E')=s}} 1=\binom{n}{i}\binom{i}{s}\binom{n-i}{i'-s}=\binom{n}{s, i-s, i'-s, n-i-i'+s},
\end{equation}
we obtain the desired result as described by Lemma \ref{Cov}.
\end{proof}

From Lemma \ref{Cov}, the variance $\sigma_{\mathrm{ud}}^2(\R_{m,n},\ep)$ can be obtained easily, which we summarize below:
\begin{theorem}\label{Varud}
\begin{eqnarray*}
 \sigma_{\mathrm{ud}}^2(\R_{m,n},\ep)
=
\sum_{i,i'=1}^m\sum_{s=1}^{\min\{i,i'\}}\psi_m(i)\psi_m(i')
\left(\frac{1}{\psi_m(s)}-1\right)\binom{n}{s, i-s, i'-s, n-i-i'+s}\ep^{i+i'}(1-\ep)^{2n-i-i'}.\end{eqnarray*}
\end{theorem}

%\subsection{Error exponent of $\sigma_{\mathrm{ud}}^2(\R_{(1-R)n,n},\ep)$}\label{Pferrvar}

For $0<R<1$, the error exponent of the variance $\sigma_{\mathrm{ud}}^2(\R_{(1-R)n,n},\ep)$ is defined by
\begin{equation}\label{Se}
S_{\mathrm{ud}}(\ep):=-\lim_{n \to \infty} \frac{1}{n}\log_q \sigma_{\mathrm{ud}}^2(\R_{(1-R)n,n},\ep),
\end{equation}
if the limit exists. We obtain:
\begin{theorem}\label{Errexpvar}
$$S_\mathrm{ud}(\ep)=\begin{cases}1-R-\log_q[1+(q-1)\ep^2] &\left(0 < R \leq \frac{1-\ep}{1+(q-1)\ep^2}\right)\\
(1-R)\log_q\left(\frac{\kappa_0}{\ep^2}\right)+R\log_q\left(\frac{2R-1+\kappa_0}{(1-\ep)^2}\right) &\left(\frac{1-\ep}{1+(q-1)\ep^2} < R < 1\right)
\end{cases}.$$
where $\kappa_0:=\kappa_0(R)$ is given by
\begin{equation}\label{k0}
\kappa_0(R):=1-R-\frac{\sqrt{4R(1-R)(q-1)+1}-1}{2(q-1)}.
\end{equation}
\end{theorem}

A plot of the function $S_{\mathrm{ud}}(\ep)$ for $q=2,\ep=0.25$ in the range $0 < R < 1$ is given by {\bf Fig. 4}.

\begin{figure}
\begin{center}
\includegraphics[angle=0,width=0.5 \textwidth,height=0.2 \textheight]{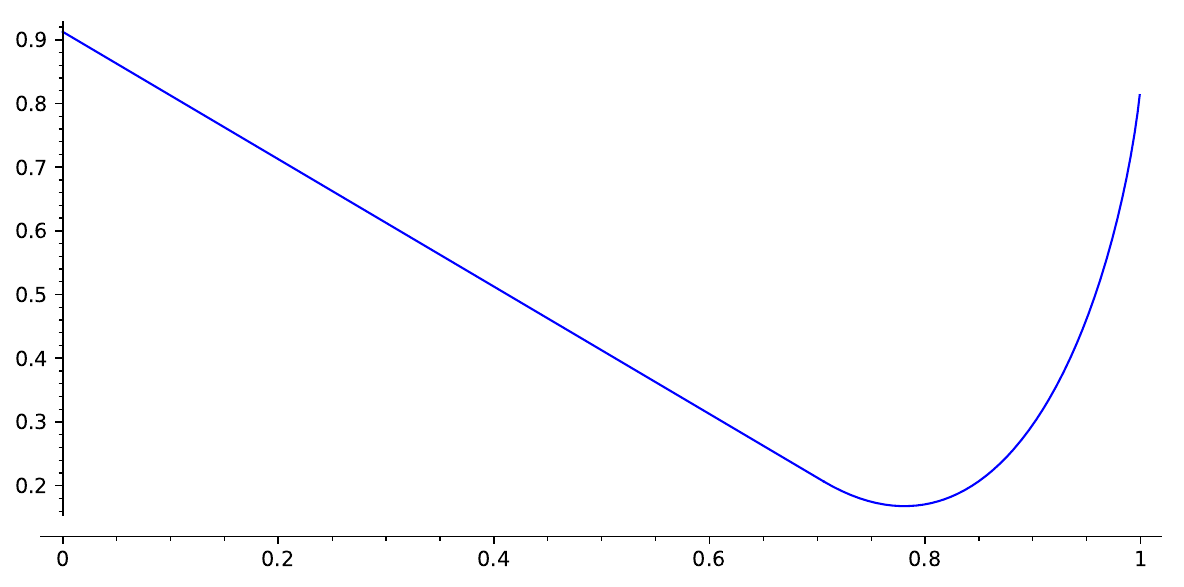}
\caption{The error exponent $S_{\mathrm{ud}}(\ep)$ for $0<R<1$, where $q=2, \ep=0.25$.}
\end{center}
\end{figure}

We remark that the proof of Theorem \ref{Errexpvar} follows a similar argument as that of Theorem \ref{Errexp}, though here the computation of $S_{\mathrm{ud}}(\ep)$ is much more complex as it involves a lot more technical manipulations. In order to streamline the idea of the proof in this section, we first assume Theorem \ref{Errexpvar} and leave its proof to Section \ref{Appen1} {\bf Appendix}. Then Theorem \ref{Concen} can be proved easily by using the standard Chebyshev's inequality.

\subsection{Proof of Theorem \ref{Concen}}

We have, by definition of error exponents,
$$P_{\mathrm{ud}}(\R_{(1-R)n,n},\ep)=q^{-n(T_\mathrm{ud}(\ep)+o(1))},$$
and
$$\sigma_{\mathrm{ud}}^2(\R_{(1-R)n,n},\ep)=q^{-n(S_\mathrm{ud}(\ep)+o(1))}.$$
These imply
$$\frac{\sigma_{\mathrm{ud}}^2(\R_{(1-R)n,n},\ep)}{P_{\mathrm{ud}}^2(\R_{(1-R)n,n},\ep)}=
\frac{q^{-n(S_\mathrm{ud}(\ep)+o(1))}}{q^{-2n(T_{\mathrm{ud}}(\ep)+o(1))}}
=q^{-n(S_\mathrm{ud}(\ep)-2T_{\mathrm{ud}}(\ep)+o(1))}.$$
The right hand side of the above will tend to 0 as $n \to \infty$ if
\begin{equation}\label{condition2}
c(\ep,R):=S_\mathrm{ud}(\ep)-2T_\mathrm{ud}(\ep) > 0.
\end{equation}
Thus under (\ref{condition2}), by Chebyshev's inequality, we have, for any fixed $\delta > 0$,
$$\P_{\R_{(1-R)n,n}}\left[\left|\frac{P_\mathrm{ud}(H,\ep)}{P_{\mathrm{ud}}(\R_{(1-R)n,n},\ep)}
-1\right| \geq \delta\right] \leq \frac{\sigma_{\mathrm{ud}}^2(\R_{(1-R)n,n},\ep)}{\delta^2P_{\mathrm{ud}}^2(\R_{(1-R)n,n},\ep)}
=q^{-n(c(\ep,R)+o(1))},$$
that is, $\frac{P_\mathrm{ud}(H,\ep)}{P_{\mathrm{ud}}(\R_{(1-R)n,n},\ep)} \to 1$ \textbf{WHP} as $n \to \infty$.

To prove Theorem \ref{Concen}, it remains to verify that (\ref{condition2}) holds true under the assumptions of either (1) or (2) of Theorem \ref{Concen}.

\textbf{Case 1.} $\frac{1-\ep}{1+(q-1)\ep^2} \le R < 1-\ep$

Since $\frac{1-\ep}{1+(q-1)\ep}<\frac{1-\ep}{1+(q-1)\ep^2}$, by a simple calculation, we have
\begin{align*}
c(\ep,R)&=(1-R)\log_q\left(\frac{\kappa_0}{\ep^2}\right)-R\log_q\left(\frac{2R-1+\kappa_0}{(1-\ep)^2}\right)
-2\left[(1-R)\log_q\left(\frac{1-R}{\ep}\right)+R\log_q\left(\frac{R}{1-\ep}\right)\right]\\
&=(1-R)\log_q\left(\frac{\kappa_0}{(1-R)^2}\right)+R\log_q\left(\frac{2R-1+\kappa_0}{R^2}\right)\\
&>(1-R)\log_q\left(\frac{(1-R)^2}{(1-R)^2}\right)+R\log_q\left(\frac{2R-1+(1-R)^2}{R^2}\right)=0,
\end{align*}
so (\ref{condition2}) always holds true in this case.

\textbf{Case 2.} $\frac{1-\ep}{1+(q-1)\ep} \le R < \frac{1-\ep}{1+(q-1)\ep^2}$

We can actually obtain a slightly more general result than (2) of Theorem \ref{Concen} as follows:
\begin{lemma} \label{5:lem3} If (\ref{condition2}) holds for $R=R_0$ where $R_0$ satisfies $0<R_0 \le \frac{1-\ep}{1+(q-1)\ep^2}$, then (\ref{condition2}) also holds true for any $R \in \left[R_0,\frac{1-\ep}{1+(q-1)\ep^2}\right]$.
\end{lemma}
\begin{proof}[Proof of Lemma \ref{5:lem3}]
First, from \textbf{Case 1} and the continuity of $S_\mathrm{ud}(\ep)$ and $T_\mathrm{ud}(\ep)$, we know that (\ref{condition2}) holds for $R=\frac{1-\ep}{1+(q-1)\ep^2}$. Hence such $R_0$ always exists.

Now if $R \in \left[\frac{1-\ep}{1+(q-1)\ep},\frac{1-\ep}{1+(q-1)\ep^2}\right]$, then we have
$$c(\ep,R)=1-R-\log_q[1+(q-1)\ep^2]-2\left[(1-R)\log_q\left(\frac{1-R}{\ep}\right)+R\log_q\left(\frac{R}{1-\ep}\right)\right].$$
Differentiating $c(\ep,R)$ twice with respect to $R$, we have
$$\frac{\partial^2 c(\ep,R)}{\partial R^2}=-\frac{2}{R(1-R)\ln q},$$
which is negative for $R \in \left[\frac{1-\ep}{1+(q-1)\ep},\frac{1-\ep}{1+(q-1)\ep^2}\right]$, and thus $c(\ep,R)$ is convex $\cap$ for $R$ within that interval.

Next, for $0 < R < \frac{1-\ep}{1+(q-1)\ep}$, we note that
\begin{align*}
c(\ep,R)&=1-R-\log_q[1+(q-1)\ep^2]-2\{1-R-\log_q[1+(q-1)\ep]\}\\
&=R-1-\log_q[1+(q-1)\ep^2]+2\log_q[1+(q-1)\ep]
\end{align*}
is a linear function in $R$ with positive slope. Hence the function $c(\ep,R)$ is convex $\cap$ for $R$ within the whole interval $\left[R_0,\frac{1-\ep}{1+(q-1)\ep^2}\right]$, and Lemma \ref{5:lem3} follows.
\end{proof}
Now we let $q \in \{2,3,4\}$. Consider $R=\frac{1-\ep}{1+(q-1)\ep}$. Under this special value,
$$Q(\ep):=c(\ep,R)=-\frac{q\ep}{1+(q-1)\ep}-\log_q[1+(q-1)\ep^2]+2\log_q[1+(q-1)\ep].$$
Differentiating $Q(\ep)$ with respect to $\ep$, we have
\begin{align*}
Q'(\ep)&=-\frac{q}{[1+(q-1)\ep]^2}-\frac{2(q-1)\ep}{[1+(q-1)\ep^2]\ln q}+\frac{2(q-1)}{[1+(q-1)\ep]\ln q}\\
&=-\frac{(q-1)(q\ln q+2q-2)}{[1+(q-1)\ep]^2[1+(q-1)\ep^2]\ln q} \left(\ep^2-\frac{2(q-2)}{q\ln q+2q-2} \, \ep-\frac{2-\frac{q \ln q}{q-1}}{q\ln q+2q-2}\right).
\end{align*}
It is easy to check that when $q \in \{2,3,4\}$, then we have
\[0 \leq \frac{2(q-2)}{q\ln q+2q-2} < 1, \quad Q'(0)>0, Q'(1)<0.\]
Therefore exactly one of the two roots for $Q'(\ep)$ lies in the desired range $0 < \ep < 1$, and $Q(\ep)$ attains local maximum at that point. Since $Q(0)=Q(1)=0$, we conclude that $Q(\ep)$ is positive for $0 < \ep < 1$, and therefore (\ref{condition2}) holds for $R=\frac{1-\ep}{1+(q-1)\ep}$.

Now by Lemma \ref{5:lem3}, we see that Equation (\ref{condition2}) holds for $R \in \left[\frac{1-\ep}{1+(q-1)\ep},\frac{1-\ep}{1+(q-1)\ep^2}\right]$ when $q=2,3,4$. This completes the proof of Theorem \ref{Concen}.

\section{Appendix: Proof of Theorem \ref{Errexpvar}}\label{Appen1}

In order to prove Theorem \ref{Errexpvar}, we first need to obtain
\begin{lemma}\label{Errexpvarlem}
$$S_{\mathrm{ud}}(\ep)=-\sup_{t,t',\kappa} g(t,t',\kappa),$$
where
\begin{equation}\label{gt}
g(t,t',\kappa)=(\kappa-1+R)+h(\kappa,t-\kappa,t'-\kappa,1-t-t'+\kappa)+(t+t')\log_q\ep+(2-t-t')\log_q(1-\ep),
\end{equation}
and the supremum is taken over positive real numbers $t,t',\kappa$ satisfying
\begin{equation}\label{range}
t,t' \leq 1-R, \quad t+t'-1 \leq \kappa \leq \min\{t,t'\}.
\end{equation}
Here
\begin{equation}\label{multien}
h(t_1,t_2,\cdots,t_r)=-\sum_{j=1}^r t_j\log_q t_j\quad \left(t_j \geq 0, \sum_{j=1}^r t_j=1\right)
\end{equation}
is the multi-entropy function (in $q$-its).
\end{lemma}
\begin{proof}[Proof of Lemma \ref{Errexpvarlem}]
Write $m=(1-R)n$. We define
$$g_{i,i',s}:=\psi_m(i)\psi_m(i')\left(\frac{1}{\psi_m(s)}-1\right)\binom{n}{s, i-s, i'-s, n-i-i'+s}\ep^{i+i'}(1-\ep)^{2n-i-i'}.$$
We note that $g_{i,i',s} \geq 0$ for all integers $i,i',s$, and is nonzero if and only if $1 \leq i,i' \leq m, 1 \leq s \leq \min\{i,i'\}$ and $i'-s \leq n-i$.
Then we can rewrite the term $\sigma_{\mathrm{ud}}^2(\R_{m,n},\ep)$ (see Theorem \ref{Varud}) as
$$\sigma_{\mathrm{ud}}^2(\R_{m,n},\ep)=\sum_{i,i',s} g_{i,i',s},$$
where summation is over all integers $i,i',s$ satisfying the conditions
\begin{equation}\label{conditionud}
1 \leq i,i' \leq m, \quad \max\{1,i+i'-n\} \leq s \leq \min\{i,i'\}.
\end{equation}
There are at most $m^3$ such integer triples $(i,i',s)$. Since $\lim_{n \to \infty} \frac{\log_q m^3}{n}=0$, we have
\begin{equation}\label{Varud2}
\frac{1}{n}\log_q\sigma_{\mathrm{ud}}^2(\R_{m,n},\ep)=o(1)+\frac{1}{n}\max_{i,i',s} \log_q g_{i,i',s}.
\end{equation}
Now we need a careful analysis of the quantity $g_{i,i',s}$. First, using $T:=q^{-1}$ and $M:=(T)_\infty \in (0,1)$, we observe that
$$M^2 \leq \psi_m(i)\psi_m(i')=
\frac{(T)_m^2}{(T)_{m-i}(T)_{m-i'}} \leq \frac{1}{M^2}, \quad \forall 1 \le i, i' \le m.$$
So
$$\frac{1}{n}\log_q[\psi_m(i)\psi_m(i')]=o(1).$$
Next, for $1 \leq s \leq m$, using (\ref{4:qtot}), we have,
\begin{align*}
\frac{1}{\psi_m(s)}-1&=\frac{(T)_{m-s}}{(T)_m}-1
=\frac{(T)_{m-s}}{(T)_m}\left(1-\prod_{j=m-s+1}^m(1-T^j)\right).
\end{align*}
It is easy to see that
\[1-T^{m-s+1} \ge \prod_{j=m-s+1}^m(1-T^j) > 1-\sum_{j=m-s+1}^{\infty}T^j=1-\frac{q}{q-1}T^{m-s+1}, \]
so we have
\[\frac{1}{n} \log_q \left(\frac{1}{\psi_m(s)}-1\right)=\frac{s-m-1}{n}+o(1).\]
Therefore we can simplify (\ref{Varud2}) as
\begin{align*}
&\quad\frac{1}{n}\log_q\sigma_{\mathrm{ud}}^2(\R_{m,n},\ep)\\
&=o(1)+\max_{i,i',s}\left[\frac{s-m-1}{n}+\frac{1}{n}\log_q\binom{n}{s, i-s, i'-s, n-i-i'+s}+\frac{i+i'}{n}\log_q \ep+\frac{2n-i-i'}{n}\log_q(1-\ep)\right].
\end{align*}
Letting $i=tn, j=t'n$ and $s=\kappa n$, using the following generalization of (\ref{entropy}) (which can be verified similarly)
$$\frac{1}{n}\log_q \binom{n}{t_1n,t_2n,\cdots,t_rn}=h(t_1,t_2,\cdots,t_r)+o(1),$$
where $h(t_1,t_2,\cdots,t_r)$ (with $\sum_{j=1}^r t_j=1$) is defined as in (\ref{multien}), and then taking $n \to \infty$, we obtain
$$S_{\mathrm{ud}}(\ep)=-\sup_{t,t',\kappa} g(t,t',\kappa),$$
where $g(t,t',\kappa)$ is defined as in (\ref{gt}) and the supremum is taken over positive real numbers $t,t',\kappa$ satisfying (\ref{range}). This completes the proof of Lemma \ref{Errexpvarlem}.
\end{proof}

Using Lemma \ref{Errexpvarlem}, now we provide a detailed proof of Theorem \ref{Errexpvar}.

We start from the function $g(t,t',\kappa)$ given in (\ref{gt}). It is best to first take the supremum over $t'$ while fixing $t$ and $\kappa$. Therefore we differentiate $g$ with respect to $t'$ while keeping $t$ and $\kappa$ fixed, where the range of $t'$ should be taken as $\kappa \leq t' \leq \min\{\kappa-t+1,1-R\}$:
$$\frac{\partial g}{\partial t'}=\log_q\left(\frac{(1-t-t'+\kappa)\ep}{(t'-\kappa)(1-\ep)}\right).$$

Solving for $\frac{\partial g}{\partial t'}=0$, we get
$t'=\ep(1-t)+\kappa$, and we check that $g$ attains local maximum at that point.

This leads us to consider two cases:

\textbf{Case 1.} $t \geq 1-\frac{1-R-\kappa}{\ep}$

In this case the critical point $\ep(1-t)+\kappa$ lies within the required range of $t'$. We see that the maximum of $g$ is
\begin{align*}
G_1(t,\kappa)&:=g(t,\ep(1-t)+\kappa,\kappa)\\
&=(\kappa-1+R)+h\left(\kappa,t-\kappa,\ep(1-t),(1-\ep)(1-t)\right)+[(\ep+\kappa)+(1-\ep)t]\log_q\ep\\
&\quad +[(2-\ep-\kappa)-(1-\ep)t]\log_q(1-\ep).
\end{align*}
Differentiating both sides with respect to $t$ while keeping $\kappa$ fixed, where the range of $t$ is taken as $\max\{\kappa,1-\frac{1-R-\kappa}{\ep}\} \leq t \leq 1-R$, we have
$$\frac{\partial G_1}{\partial t}=\log_q\left(\frac{\ep(1-t)}{t-\kappa}\right).$$
Solving for $\frac{\partial G_1}{\partial t}=0$, we get $t=\frac{\ep+\kappa}{1+\ep}$. We also see that $G_1$ attains local maximum at this point. We note that $\kappa \leq \frac{\ep+\kappa}{1+\ep}$ as $\kappa < 1$. In order for the range of $t$ to be nonempty, we also need $1-\frac{1-R-\kappa}{\ep} \leq 1-R$, which is true if and only if $0 < \kappa \leq 1-R-R\ep$ (and this range is nonempty if and only if $0 < R < \frac{1}{1+\ep}$). We then obtain
$$1-\frac{1-R-\kappa}{\ep} \leq \frac{\ep+\kappa}{1+\ep} \leq 1-R,$$
so the maximum of $g$ is
\begin{align*}
\mathcal{G}_1(\kappa)&:=G_1\left(\frac{\ep+\kappa}{1+\ep},\kappa\right)\\
&=(\kappa-1+R)+h\left(\kappa,\frac{\ep(1-\kappa)}{1+\ep},\frac{\ep(1-\kappa)}{1+\ep},\frac{(1-\ep)(1-\kappa)}{1+\ep}\right)+2\left(\frac{\ep+\kappa}{1+\ep}\right)\log_q\ep+2\left(\frac{1-\kappa}{1+\ep}\right)\log_q(1-\ep).
\end{align*}
Differentiating both sides with respect to $\kappa$ and checking for $\mathcal{G}_1'(\kappa)=0$, we have $\kappa=\frac{q\ep^2}{1+(q-1)\ep^2}$. In addition, $\mathcal{G}_1$ attains local maximum at this critical point.

Then we have the following two possible cases:
\begin{enumerate}
\item $0 < R \leq \frac{1-\ep}{1+(q-1)\ep^2}$

In this case $\frac{q\ep^2}{1+(q-1)\ep^2} \leq 1-R-R\ep$, and so the maximum is
$$\mathcal{G}_1\left(\frac{q\ep^2}{1+(q-1)\ep^2}\right)=-1+R+\log_q[1+(q-1)\ep^2].$$
\item $\frac{1-\ep}{1+(q-1)\ep^2} < R < \frac{1}{1+\ep}$

In this case $0 < 1-R-R\ep < \frac{q\ep^2}{1+(q-1)\ep^2}$, and so the maximum is
$$\mathcal{G}_1(1-R-R\ep)=-R\ep-(1-R-R\ep)\log_q\left(\frac{1-R-R\ep}{\ep^2}\right)-R(1+\ep)\log_q\left(\frac{R}{1-\ep}\right).$$
\end{enumerate}
\textbf{Case 2.} $t < 1-\frac{1-R-\kappa}{\ep}$

In this case the critical point $\ep(1-t)+\kappa$ is larger than the upper bound $1-R$ of $t'$. Hence maximum is
\begin{align*}
G_2(t,\kappa)&:=g(t,1-R,\kappa)\\
&=(\kappa-1+R)+h(\kappa,t-\kappa,1-R-\kappa,R-t+\kappa)+(t+1-R)\log_q\ep+(1+R-t)\log_q(1-\ep).
\end{align*}
Differentiating both sides with respect to $t$ while keeping $\kappa$ fixed, where the range of $t$ is taken as $\kappa \leq t \leq \min\{1-\frac{1-R-\kappa}{\ep}, 1-R\}$, we get
$$\frac{\partial G_2}{\partial t}=\log_q\left(\frac{\ep(R+\kappa-t)}{(t-\kappa)(1-\ep)}\right).$$
Solving for $\frac{\partial G_2}{\partial t}=0$, we have $t=R\ep+\kappa$. We also see that $G_2$ attains local maximum at this point. We already see that $\kappa \leq R\ep+\kappa$. Hence we need to compare the critical point with $\min\{1-\frac{1-R-\kappa}{\ep}, 1-R\}$. First in order that the range of $t$ is nonempty, we must have
$\kappa \leq \min\{1-\frac{1-R-\kappa}{\ep},1-R\}$, which is true if and only if $1-\frac{R}{1-\ep} \leq \kappa \leq 1-R$. We can further divide into two subcases:

\textbf{Case 2a.} $1-\frac{R}{1-\ep} \leq \kappa \leq 1-R-R\ep$

In this case we have $1-\frac{1-R-\kappa}{\ep} \leq R\ep+\kappa \leq 1-R$. Hence the maximum occurs at $t=1-\frac{1-R-\kappa}{\ep}$. Note that this value of $t$ is precisely the one in which $t'=1-R=\ep(1-t)+\kappa$. Therefore it is covered in \textbf{Case 1} already, and the maximum so obtained cannot be larger than the value calculated in that case. Note that this case can only happen when $0 < R  < \frac{1}{1+\ep}$.

\textbf{Case 2b.} $\max\{0, 1-R-R\ep\} < \kappa \leq 1-R$

In this case we have $1-R < R\ep+\kappa < 1-\frac{1-R-\kappa}{\ep}$. Hence the maximum occurs at $t=1-R$.

Then we get
\begin{align*}
\mathcal{G}_2(\kappa)&:=G_2(1-R,\kappa)\\
&=(\kappa-1+R)+h(\kappa,1-R-\kappa,1-R-\kappa,2R-1+\kappa)+2(1-R)\log_q\ep+2R\log_q(1-\ep).
\end{align*}
Differentiating both sides with respect to $\kappa$, we have
\begin{equation}\label{G2}
\mathcal{G}_2'(\kappa)=\log_q\left(\frac{q(1-R-\kappa)^2}{\kappa(2R-1+\kappa)}\right).
\end{equation}
Solving for $\mathcal{G}_2'(\kappa)=0$, we obtain two roots
$$\kappa_{\pm}=1-R+\frac{1 \pm \sqrt{4R(1-R)(q-1)+1}}{2(q-1)}.$$
However under our assumption we require $\kappa \leq 1-R$. It is easy to see that we should then take $\kappa_-$. This is precisely $\kappa_0$ given by (\ref{k0}).

Note that the number
$$N(\kappa)=\frac{q(1-R-\kappa)^2}{\kappa(2R-1+\kappa)}$$
inside the logarithm in right hand side of (\ref{G2}) is a strictly decreasing function within our range of $\kappa$. Hence it suffices to check the value of $N(\kappa)$ at the two bounds to see whether $\kappa_0$ is within our range (in particular $N(\kappa_0)=1$). It is clear that
$$N(1-R)=0,\quad N(1-R-R\ep)=\frac{qR\ep^2}{(1-\ep)(1-R-R\ep)}$$
and
$$\lim_{\kappa \to 0^+} N(\kappa)=+\infty.$$
Then we have two cases again:
\begin{enumerate}
\item $0 < R \leq \frac{1-\ep}{1+(q-1)\ep^2}$

In this case we have $N(1-R-R\ep) \leq 1$, and so $\mathcal{G}_2'(\kappa) \leq 0$ within the range of $\kappa$. This shows that the maximum occurs at $\kappa=1-R-R\ep$. Note that this also implies $t=1-R=1-\frac{1-R-\kappa}{\ep}$. Since this value is already covered in \textbf{Case 2a}, the maximum cannot be greater than in that case and thereby in \textbf{Case 1} too.

\item $\frac{1-\ep}{1+(q-1)\ep^2} < R < 1$

In this case we have $N(1-R-R\ep) > 1$, so one of the critical points (in fact the smaller one) is within our range. That number is exactly $\kappa_0$ defined in (\ref{k0}). Since $\mathcal{G}_2'(\kappa)$ is decreasing in our range, this implies $\mathcal{G}_2$ attains maximum at $\kappa_0$. The maximum will then be
\begin{align*}
\mathcal{G}_2(\kappa_0)&=(\kappa_0-1+R)+h(\kappa_0,1-R-\kappa_0,1-R-\kappa_0,2R-1+\kappa_0)+2(1-R)\log_q\ep+2R\log_q(1-\ep)\\
&=-(1-R)\log_q\left(\frac{\kappa_0}{\ep^2}\right)-R\log_q\left(\frac{2R-1+\kappa_0}{(1-\ep)^2}\right)
\end{align*}
after simplification and applying the relation $\mathcal{G}_2'(\kappa_0)=0$.
\end{enumerate}
Note that in particular when $\frac{1-\ep}{1+(q-1)\ep^2} < R < \frac{1}{1+\ep}$, this value is larger than $\mathcal{G}_2(1-R-R\ep)=-R\ep-(1-R-R\ep)\log_q\left(\frac{1-R-R\ep}{\ep^2}\right)-R(1+\ep)\log_q\left(\frac{R}{1-\ep}\right)=\mathcal{G}_1(1-R-R\ep)$.

Combining all above cases, we finally obtain
$$S_\mathrm{ud}(\ep)=-\sup_{t,t',\kappa}g(t,t',\kappa)=\begin{cases}1-R-\log_q[1+(q-1)\ep^2] &\left(0 < R \leq \frac{1-\ep}{1+(q-1)\ep^2}\right)\\
(1-R)\log_q\left(\frac{\kappa_0}{\ep^2}\right)+R\log_q\left(\frac{2R-1+\kappa_0}{(1-\ep)^2}\right) &\left(\frac{1-\ep}{1+(q-1)\ep^2} < R < 1\right).
\end{cases}$$
This completes the proof of Theorem \ref{Errexpvar}.

\section{Conclusion}\label{Conclusion}
In this paper we carried out an in-depth study on the average decoding error probabilities of the random parity-check matrix ensemble $\R_{m,n}$ over the erasure channel under three decoding principles, namely \emph{unambiguous decoding, maximum likelihood decoding} and \emph{list decoding}.

\begin{enumerate}
\item[(1).] We obtained explicit formulas for the average decoding error probabilities of the ensemble under these three decoding principles and computed the error exponents. We also compare the results with the random $[n,k]_q$ code ensemble studied before. 

\item[(2).] For \emph{unambiguous decoding}, we computed the variance of the decoding error probability of the ensemble and the error exponent of the variance, from which we derived a strong concentration result, that is, under general conditions, the ratio of the decoding error probability of a random code in the ensemble and the average decoding error probability of the ensemble converges to 1 with high probability when the code length goes to infinity.

\end{enumerate}

It might be interesting to extend the results of (2) to general \emph{list decoding} and \emph{maximum likelihood decoding} for the ensemble $\R_{m,n}$. As it turns out, the variance of decoding error probability in these two cases can still be computed, but the expressions are much more complicated and it is difficult to obtain explicit formulas for their error exponents, and hence a concentration result from them. We leave this as an open question for future research. 

It may be also interesting to carry out such computation for the random $[n,k]_q$ code ensemble.

%\input{bib-V03-22}

%\begin{IEEEbiography}[{\includegraphics[width=1in,height=1.20in,clip,keepaspectratio]{2010Xiong.jpg}}]
%{Maosheng Xiong} received the Ph.D degree in mathematics from University of Illinois at Urbana-Champaign in 2007. He was a postdoc at Pennsylvania State University from August 2007 to June 2010. He joined the math department of Hong Kong University of Science and Technology in July 2010 and is currently an assistant professor. His area of research interests is in coding theory and information theory.
%\end{IEEEbiography}


\begin{thebibliography}{99}

\bibitem{Berk} E. R. Berlekamp: The technology of error-correcting codes. In: Proceedings of the IEEE \textbf{68}(5), 564--593 (1980).
    
\bibitem{Byers} J. W. Byers, M. Luby, M. Mitzenmacher and A. Rege: A digital fountain approach to reliable distribution of bulk data. In: Proc. ACM SIGCOMM Conf. Appl., Technol., pp. 56--67. Architectures, Protocols Comput. Commun., Vancouver, BC, Canada (1998).

\bibitem{InfThe} T. M. Cover, J. A. Thomas: Elements of Information Theory, 1st Ed. Wiley-Interscience, New York, NY, USA (1991).

\bibitem{Di} C. Di, D. Proietti, I. E. Telatar, T. J. Richardson, R. L. Urbanke: `Finite-length analysis of low-density parity-check codes on the binary erasure channel. Special issue on Shannon theory: perspective, trends, and applications. IEEE Trans. Inform. Theory \textbf{48}(6), 1570--1579 (2002).

\bibitem{Didier} F. Didier: A new upper bound on the block error probability after decoding over the erasure channel. IEEE Trans. Inform. Theory \textbf{52}(10), 4496--4503 (2006).

\bibitem{Fashandi} S. Fashandi, S. O. Gharan, and A. K. Khandani: Coding over an erasure
channel with a large alphabet size. In: Proc. 2008 IEEE Int. Symp. Inform. Theory, pp. 1053--1057.

\bibitem{Gallager2} R. G. Gallager: Low Density Parity Check Codes.  MIT Press (1963); IRE Trans. IT-8, pp. 21--28, Cambridge, MA (1962). 

\bibitem{Gallager} R. G. Gallager: Information Theory and Reliable Communication. Wiley, New York, NY, USA (1968).

\bibitem{Lemes} L. C. Lemes, M. Firer: Generalized weights and bounds for error probability over erasure channels. In: Proc. Inf. Theory Appl. Workshop (ITA), pp. 108. San Diego, CA, USA (2014).

\bibitem{Liva1} G. Liva, E. Paolini, M. Chiani: Bounds on the error probability of block codes over the $q$-ary erasure channel. IEEE Trans. Inform. Theory \textbf{61}(6), 2156--2165 (2013).

\bibitem{Luby} M. Luby, M. Mitzenmacher, A. Shokrollahi, D. A. Sipelman, V. Stemann: Practical loss-recilient codes. In: Proc. 29th Annual ACM Symp. Theory Comput., pp. 150--159 (1997).

\bibitem{Lun} D. S. Lun, M. M\'{e}dard, R. Koetter, M. Effros: On coding for reliable communication over packet networks. Phys. Commun. \textbf{1}(1), 3--20 (2008).

\bibitem{MacWilliams} F. J. MacWilliams, N. J. A. Sloane: The Theory of Error-Correcting Codes. North-Holland Mathematical Library \textbf{16}, Amsterdam, The Netherlands (1981).

\bibitem{FFW} L. Shen, F.-W. Fu: The decoding error probability of linear codes over the erasure channel. IEEE Trans. Inform. Theory \textbf{65}(10), 6194--6203 (2019).

\bibitem{Wadayama} T. Wadayama: On the undetected error probability of binary matrix ensemblesl. IEEE Trans. Inform. Theory \textbf{56}(5), 2168--2176 (2010).

\bibitem{Weber} J. H. Weber, K. A. S. Abdel-Ghaffar: Results on parity-check matrices with optimal stopping and/or dead-end set enumerators. IEEE Trans. Inform. Theory \textbf{54}(3), 1368--1374 (2008).

\bibitem{Viterbi} A. J. Viterbi, J. K. Omura: Principles of Digitial Communication and Coding. McGraw-Hill, New York, NY, USA (1979).

\bibitem{Xio} M. Xiong: Decoding error probability in the erasure channel and the $r$-th support weight distribution (in Chinese). Sci Sin Math \textbf{51}, 1--14 (2021). doi: 10.1360/SSM-2021-0019. 

%\bibitem{Yeung} R. W. Yeung, \emph{Information Theory and Network Coding}. New York, NY, USA: Springer Science+Business Media, LLC, 2008.
\end{thebibliography}
\end{document}